\documentclass[10pt,conference]{IEEEtran}
\usepackage[utf8]{inputenc}
\usepackage[T1]{fontenc}
\usepackage{amsmath}
\usepackage{amssymb}
\usepackage{amsthm}
\usepackage[caption=false]{subfig}
\usepackage{tikz}
\usetikzlibrary{decorations.pathreplacing,decorations.pathmorphing}
\usepackage{xspace}

\newtheorem{construction}{Construction}
\newtheorem{example}{Example}
\newtheorem{lemma}{Lemma}
\newtheorem{theorem}{Theorem}

\newcommand{\gnot}{\ensuremath{\text{\small\normalfont{\textsf{NOT}}}}\xspace}
\newcommand{\cnot}{\ensuremath{\text{\small\normalfont{\textsf{CNOT}}}}\xspace}
\newcommand{\ccnot}{\ensuremath{\text{\small\normalfont{\textsf{CCNOT}}}}\xspace}
\newcommand{\cand}{\ensuremath{\text{\small\normalfont{\textsf{AND}}}}\xspace}
\newcommand{\mone}{\overline{1}}

\begin{document}

\title{Quantum Circuits for \\ Functionally Controlled NOT Gates}
\author{%
  \IEEEauthorblockN{Mathias Soeken \quad Martin Roetteler}
  \IEEEauthorblockA{Microsoft Quantum, Redmond, United States}
}

\maketitle

\begin{abstract}
  We generalize quantum circuits for the Toffoli gate presented by
  Selinger~\cite{Selinger13} and Jones~\cite{Jones13} for functionally
  controlled NOT gates, i.e., $X$ gates controlled by arbitrary $n$-variable
  Boolean functions. Our constructions target the gate set consisting of
  Clifford gates and single qubit rotations by arbitrary angles. Our
  constructions use the Walsh-Hadamard spectrum of Boolean functions and build
  on the work by Schuch and Siewert~\cite{SS03} and Welch et al~\cite{WGMA14}.
  We present quantum circuits for the case where the target qubit is in an
  arbitrary state as well as the special case where the target is in a known
  state. Additionally, we present constructions that require no auxiliary qubits
  and constructions that have a rotation depth of 1.
\end{abstract}

\section{Background and Motivation}

Finding quantum circuit implementations of subroutines that are given as
classical functions is a problem that occurs in many contexts. Examples includes
Shor's algorithm for factoring and dlogs \cite{Shor97}, Grover's search
algorithm \cite{Grover96}, quantum walk algorithms \cite{Kempe03}, the
Harrow-Hassidim-Lloyd algorithm for solving linear equations \cite{HHL09,CJS13},
and quantum simulation methods \cite{BCK15,BCC+14}. In this paper, we consider
the case where the classical functions are provided as a Boolean function, i.e.,
we are interested in finding quantum circuits that implement the unitary
\begin{equation}
  U_f:|x\rangle|y\rangle|0^\ell\rangle \mapsto
  |x\rangle|y\oplus f(x)\rangle|0^\ell\rangle,
\end{equation}
where $f(x)$ is a Boolean function over $n$ variables
$x = x_1, \dots, x_n$. 

As the target gate set we consider the universal gate set Clifford+$R_1$, which
is generated by the \cnot gate, the Hadamard gate $H =
\tfrac{1}{\sqrt{2}}\left(\begin{smallmatrix} 1 & 1 \\ 1 & -1
\end{smallmatrix}\right)$, and unitaries $R_1(\frac{k\pi}{2^j}) =
\operatorname{diag}(1, e^{\mathrm{i}k\pi2^{-j}})$ as well as their adjoints for
arbitrary nonnegative integers $j$ and $k$.  We also allow classical control
based on intermediate measurement outcomes.  The well-known Clifford+$T$ gate
library is a special case in which $j = 2$, since $T = R_1(\frac{\pi}{4})$ and
$T^\dagger = R_1^\dagger(\frac{\pi}{4}) = R_1(-\frac{\pi}{4})$.  A
\emph{rotation stage} is a set of $R_1$ gates in a circuit that can be executed
in parallel, and the \emph{rotation depth} of a circuit is the smallest number
of rotation stages in it.  Note that for some angles, an $R_1$ gate is a
Clifford gate, e.g., $R_1(\frac{\pi}{2}) = S$ and $R_1(\pi) = Z$; such gates are
ignored in the depth computation.  

Our main result are six different constructions for the implementation of $f$ ,
depending on different contexts of the target qubit as well as trade-offs
between circuit depth and number of qubits. Specifically, we distinguish the
cases in which $|y\rangle$ is an arbitrary quantum state, the case in which
$|y\rangle = |0\rangle$, and the case in which $|y\rangle = |f(x)\rangle$.  For
each of these three cases we show one construction where $\ell=0$, i.e., no
auxiliary qubits are required, but the circuits have rotation depth $O(2^n)$,
and we show one construction where $\ell = O(2^n)$, but the resulting circuits
have rotation depth 1.

Several previous works have presented constructions for $U_f$, in particular for
the special case in which $f(x_1, x_2) = x_1 \land x_2$, for which $U_f$ is also
referred to as Toffoli gate.  We refer to $U_f$ in this case as \ccnot.  For
instance, Selinger has shown a construction for \ccnot with rotation depth 1 and
$\ell = 4$~\cite{Selinger13}, depicted in Fig.~\ref{fig:ccnot}(a). Moreover,
Jones has shown that fewer rotation gates are required to implement \ccnot when
the target is in state $|0\rangle$ or $|f(x)\rangle$~\cite{Jones13}; we refer to
the operation in these cases as $\cand$ and $\cand^\dagger$, respectively.  The
quantum circuits are shown in Figs.~\ref{fig:ccnot}(c) and (d) based on the
constructions by Gidney~\cite{Gidney18}; also using the diagrammatic notation
for the \cand and $\cand^\dagger$ gates proposed in that reference.  A
measurement operation is required in the latter case.  The constructions for
these two cases can be leveraged to realize a \ccnot with an arbitrary target
state and $\ell=1$~\cite{Jones13}.  Schuch and Siewert presented a constructive
algorithm to find quantum circuits that perform arbitrary controlled phase-shift
operations $U_{\vec\theta} : |x\rangle \mapsto e^{-\mathrm{i}\theta_x}|x\rangle$
on $n$ qubits using only \cnot and $R_z$ gates, where $\vec\theta$ is a
$2^n$-element column vector of real rotation angles~\cite{SS03}. This operation
generalizes the functionally controlled $Z$ gate, for which all values in
$\vec\theta$ are either $0$ or $\pi$. Constructions for the functionally
controlled $Z$ gate can be used to implement functionally controlled \gnot gates
by surrounding the target qubit with $H$ gates.  The authors use the Hadamard
transform to translate the input vector $\vec\theta$ into rotation angles for
the $R_z$ gates in the circuit.  The construction has been rediscovered by Welch
et al.~\cite{WGMA14}, and improved by using less \cnot gates by exploiting Gray
codes.  An explicit construction for functionally controlled \gnot gates is
described in~\cite{SMSM19}.  Constructions that exploit the fact that the target
qubit is in state $|0\rangle$ or $|f(x)\rangle$ and aim at finding a good
trade-off between number of Toffoli gates and the number of auxiliary qubits
$\ell$ is presented in~\cite{BGM+19}.

\begin{figure}[t]
  \centering
  \subfloat[\ccnot with no auxiliary qubits]{\tikzpicture[scale=1.000000,x=1pt,y=1pt]
\filldraw[color=white] (0.000000, -7.000000) rectangle (146.000000, 35.000000);
\footnotesize
\draw[color=black] (0.000000,28.000000) -- (146.000000,28.000000);
\draw[color=black] (0.000000,28.000000) node[left] {$|x_1\rangle$};
\draw[color=black] (0.000000,14.000000) -- (146.000000,14.000000);
\draw[color=black] (0.000000,14.000000) node[left] {$|x_2\rangle$};
\draw[color=black] (0.000000,0.000000) -- (146.000000,0.000000);
\draw[color=black] (0.000000,0.000000) node[left] {$|y\rangle$};
\scope
\draw[fill=white] (8.000000, -0.000000) +(-45.000000:8.485281pt and 8.485281pt) -- +(45.000000:8.485281pt and 8.485281pt) -- +(135.000000:8.485281pt and 8.485281pt) -- +(225.000000:8.485281pt and 8.485281pt) -- cycle;
\clip (8.000000, -0.000000) +(-45.000000:8.485281pt and 8.485281pt) -- +(45.000000:8.485281pt and 8.485281pt) -- +(135.000000:8.485281pt and 8.485281pt) -- +(225.000000:8.485281pt and 8.485281pt) -- cycle;
\draw (8.000000, -0.000000) node {$H$};
\endscope
\scope
\draw[fill=white] (24.000000, 28.000000) +(-45.000000:8.485281pt and 8.485281pt) -- +(45.000000:8.485281pt and 8.485281pt) -- +(135.000000:8.485281pt and 8.485281pt) -- +(225.000000:8.485281pt and 8.485281pt) -- cycle;
\clip (24.000000, 28.000000) +(-45.000000:8.485281pt and 8.485281pt) -- +(45.000000:8.485281pt and 8.485281pt) -- +(135.000000:8.485281pt and 8.485281pt) -- +(225.000000:8.485281pt and 8.485281pt) -- cycle;
\draw (24.000000, 28.000000) node {{$T$}};
\endscope
\scope
\draw[fill=white] (24.000000, 14.000000) +(-45.000000:8.485281pt and 8.485281pt) -- +(45.000000:8.485281pt and 8.485281pt) -- +(135.000000:8.485281pt and 8.485281pt) -- +(225.000000:8.485281pt and 8.485281pt) -- cycle;
\clip (24.000000, 14.000000) +(-45.000000:8.485281pt and 8.485281pt) -- +(45.000000:8.485281pt and 8.485281pt) -- +(135.000000:8.485281pt and 8.485281pt) -- +(225.000000:8.485281pt and 8.485281pt) -- cycle;
\draw (24.000000, 14.000000) node {{$T$}};
\endscope
\scope
\draw[fill=white] (24.000000, -0.000000) +(-45.000000:8.485281pt and 8.485281pt) -- +(45.000000:8.485281pt and 8.485281pt) -- +(135.000000:8.485281pt and 8.485281pt) -- +(225.000000:8.485281pt and 8.485281pt) -- cycle;
\clip (24.000000, -0.000000) +(-45.000000:8.485281pt and 8.485281pt) -- +(45.000000:8.485281pt and 8.485281pt) -- +(135.000000:8.485281pt and 8.485281pt) -- +(225.000000:8.485281pt and 8.485281pt) -- cycle;
\draw (24.000000, -0.000000) node {{$T$}};
\endscope
\draw (37.000000,28.000000) -- (37.000000,0.000000);
\filldraw (37.000000, 28.000000) circle(1.500000pt);
\scope
\draw[fill=white] (37.000000, 14.000000) circle(3.000000pt);
\clip (37.000000, 14.000000) circle(3.000000pt);
\draw (34.000000, 14.000000) -- (40.000000, 14.000000);
\draw (37.000000, 11.000000) -- (37.000000, 17.000000);
\endscope
\scope
\draw[fill=white] (37.000000, 0.000000) circle(3.000000pt);
\clip (37.000000, 0.000000) circle(3.000000pt);
\draw (34.000000, 0.000000) -- (40.000000, 0.000000);
\draw (37.000000, -3.000000) -- (37.000000, 3.000000);
\endscope
\scope
\draw[fill=white] (50.000000, 14.000000) +(-45.000000:8.485281pt and 8.485281pt) -- +(45.000000:8.485281pt and 8.485281pt) -- +(135.000000:8.485281pt and 8.485281pt) -- +(225.000000:8.485281pt and 8.485281pt) -- cycle;
\clip (50.000000, 14.000000) +(-45.000000:8.485281pt and 8.485281pt) -- +(45.000000:8.485281pt and 8.485281pt) -- +(135.000000:8.485281pt and 8.485281pt) -- +(225.000000:8.485281pt and 8.485281pt) -- cycle;
\draw (50.000000, 14.000000) node {{$T^\dagger$}};
\endscope
\scope
\draw[fill=white] (50.000000, -0.000000) +(-45.000000:8.485281pt and 8.485281pt) -- +(45.000000:8.485281pt and 8.485281pt) -- +(135.000000:8.485281pt and 8.485281pt) -- +(225.000000:8.485281pt and 8.485281pt) -- cycle;
\clip (50.000000, -0.000000) +(-45.000000:8.485281pt and 8.485281pt) -- +(45.000000:8.485281pt and 8.485281pt) -- +(135.000000:8.485281pt and 8.485281pt) -- +(225.000000:8.485281pt and 8.485281pt) -- cycle;
\draw (50.000000, -0.000000) node {{$T^\dagger$}};
\endscope
\draw (63.000000,28.000000) -- (63.000000,14.000000);
\filldraw (63.000000, 28.000000) circle(1.500000pt);
\scope
\draw[fill=white] (63.000000, 14.000000) circle(3.000000pt);
\clip (63.000000, 14.000000) circle(3.000000pt);
\draw (60.000000, 14.000000) -- (66.000000, 14.000000);
\draw (63.000000, 11.000000) -- (63.000000, 17.000000);
\endscope
\draw (73.000000,14.000000) -- (73.000000,0.000000);
\filldraw (73.000000, 14.000000) circle(1.500000pt);
\scope
\draw[fill=white] (73.000000, 0.000000) circle(3.000000pt);
\clip (73.000000, 0.000000) circle(3.000000pt);
\draw (70.000000, 0.000000) -- (76.000000, 0.000000);
\draw (73.000000, -3.000000) -- (73.000000, 3.000000);
\endscope
\scope
\draw[fill=white] (86.000000, -0.000000) +(-45.000000:8.485281pt and 8.485281pt) -- +(45.000000:8.485281pt and 8.485281pt) -- +(135.000000:8.485281pt and 8.485281pt) -- +(225.000000:8.485281pt and 8.485281pt) -- cycle;
\clip (86.000000, -0.000000) +(-45.000000:8.485281pt and 8.485281pt) -- +(45.000000:8.485281pt and 8.485281pt) -- +(135.000000:8.485281pt and 8.485281pt) -- +(225.000000:8.485281pt and 8.485281pt) -- cycle;
\draw (86.000000, -0.000000) node {{$T$}};
\endscope
\draw (99.000000,28.000000) -- (99.000000,0.000000);
\filldraw (99.000000, 28.000000) circle(1.500000pt);
\scope
\draw[fill=white] (99.000000, 0.000000) circle(3.000000pt);
\clip (99.000000, 0.000000) circle(3.000000pt);
\draw (96.000000, 0.000000) -- (102.000000, 0.000000);
\draw (99.000000, -3.000000) -- (99.000000, 3.000000);
\endscope
\scope
\draw[fill=white] (112.000000, -0.000000) +(-45.000000:8.485281pt and 8.485281pt) -- +(45.000000:8.485281pt and 8.485281pt) -- +(135.000000:8.485281pt and 8.485281pt) -- +(225.000000:8.485281pt and 8.485281pt) -- cycle;
\clip (112.000000, -0.000000) +(-45.000000:8.485281pt and 8.485281pt) -- +(45.000000:8.485281pt and 8.485281pt) -- +(135.000000:8.485281pt and 8.485281pt) -- +(225.000000:8.485281pt and 8.485281pt) -- cycle;
\draw (112.000000, -0.000000) node {{$T^\dagger$}};
\endscope
\draw (125.000000,14.000000) -- (125.000000,0.000000);
\filldraw (125.000000, 14.000000) circle(1.500000pt);
\scope
\draw[fill=white] (125.000000, 0.000000) circle(3.000000pt);
\clip (125.000000, 0.000000) circle(3.000000pt);
\draw (122.000000, 0.000000) -- (128.000000, 0.000000);
\draw (125.000000, -3.000000) -- (125.000000, 3.000000);
\endscope
\scope
\draw[fill=white] (138.000000, -0.000000) +(-45.000000:8.485281pt and 8.485281pt) -- +(45.000000:8.485281pt and 8.485281pt) -- +(135.000000:8.485281pt and 8.485281pt) -- +(225.000000:8.485281pt and 8.485281pt) -- cycle;
\clip (138.000000, -0.000000) +(-45.000000:8.485281pt and 8.485281pt) -- +(45.000000:8.485281pt and 8.485281pt) -- +(135.000000:8.485281pt and 8.485281pt) -- +(225.000000:8.485281pt and 8.485281pt) -- cycle;
\draw (138.000000, -0.000000) node {$H$};
\endscope
\draw[color=black] (146.000000,28.000000) node[right] {$|x_1\rangle$};
\draw[color=black] (146.000000,14.000000) node[right] {$|x_2\rangle$};
\draw[color=black] (146.000000,0.000000) node[right] {${|y\oplus x_1x_2\rangle}$};
\endtikzpicture}

  \subfloat[\ccnot with rotation depth 1~\cite{Selinger13}]{\input{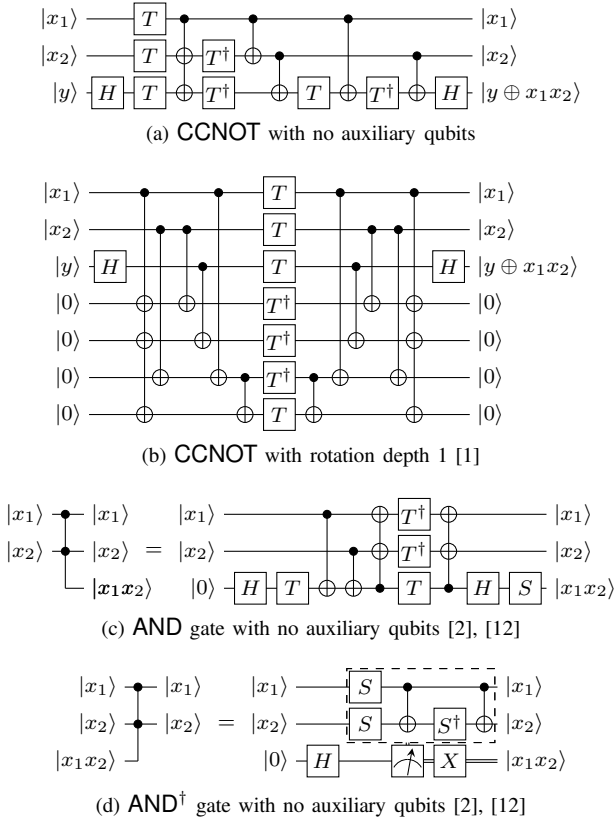}}

  \subfloat[\cand gate with no auxiliary qubits~\cite{Jones13,Gidney18}]{\tikzpicture[scale=1.000000,x=1pt,y=1pt]
\filldraw[color=white] (0.000000, -7.000000) rectangle (187.000000, 35.000000);
\footnotesize
\draw[color=black] (0.000000,28.000000) -- (19.500000,28.000000);
\draw[color=black] (57.500000,28.000000) -- (187.000000,28.000000);
\draw[color=black] (0.000000,28.000000) node[left] {$|x_1\rangle$};
\draw[color=black] (0.000000,14.000000) -- (19.500000,14.000000);
\draw[color=black] (57.500000,14.000000) -- (187.000000,14.000000);
\draw[color=black] (0.000000,14.000000) node[left] {$|x_2\rangle$};
\draw[color=black,color=white] (0.000000,0.000000) -- (5.000000,0.000000);
\draw[color=black,color=black] (5.000000,0.000000) -- (19.500000,0.000000);
\draw[color=black,color=black] (57.500000,0.000000) -- (187.000000,0.000000);
\draw[color=black] (0.000000,0.000000) node[left] {${}$};
\draw (5.000000,28.000000) -- (5.000000,0.000000);
\filldraw (5.000000, 28.000000) circle(1.500000pt);
\filldraw (5.000000, 14.000000) circle(1.500000pt);
\scope
\endscope
\scope
\endscope
\draw[color=black] (12.000000,28.000000) node[fill=white,right,minimum height=14.000000pt,minimum width=15.000000pt,inner sep=0pt] {\phantom{$|x_1\rangle$}};
\draw[color=black] (12.000000,28.000000) node[right] {$|x_1\rangle$};
\draw[color=black] (12.000000,14.000000) node[fill=white,right,minimum height=14.000000pt,minimum width=15.000000pt,inner sep=0pt] {\phantom{$|x_2\rangle$}};
\draw[color=black] (12.000000,14.000000) node[right] {$|x_2\rangle$};
\draw[color=black] (12.000000,0.000000) node[fill=white,right,minimum height=14.000000pt,minimum width=15.000000pt,inner sep=0pt] {\phantom{$|x_1x_2\rangle$}};
\draw[color=black] (12.000000,0.000000) node[right] {$|x_1x_2\rangle$};
\draw[fill=white,color=white] (31.000000, -6.000000) rectangle (46.000000, 34.000000);
\draw (38.500000, 14.000000) node {$=$};
\draw[color=black] (65.000000,28.000000) node[fill=white,left,minimum height=14.000000pt,minimum width=15.000000pt,inner sep=0pt] {\phantom{$|x_1\rangle$}};
\draw[color=black] (65.000000,28.000000) node[left] {$|x_1\rangle$};
\draw[color=black] (65.000000,14.000000) node[fill=white,left,minimum height=14.000000pt,minimum width=15.000000pt,inner sep=0pt] {\phantom{$|x_2\rangle$}};
\draw[color=black] (65.000000,14.000000) node[left] {$|x_2\rangle$};
\draw[color=black] (65.000000,0.000000) node[fill=white,left,minimum height=14.000000pt,minimum width=15.000000pt,inner sep=0pt] {\phantom{$|0\rangle$}};
\draw[color=black] (65.000000,0.000000) node[left] {$|0\rangle$};
\scope
\draw[fill=white] (75.000000, -0.000000) +(-45.000000:8.485281pt and 8.485281pt) -- +(45.000000:8.485281pt and 8.485281pt) -- +(135.000000:8.485281pt and 8.485281pt) -- +(225.000000:8.485281pt and 8.485281pt) -- cycle;
\clip (75.000000, -0.000000) +(-45.000000:8.485281pt and 8.485281pt) -- +(45.000000:8.485281pt and 8.485281pt) -- +(135.000000:8.485281pt and 8.485281pt) -- +(225.000000:8.485281pt and 8.485281pt) -- cycle;
\draw (75.000000, -0.000000) node {$H$};
\endscope
\scope
\draw[fill=white] (91.000000, -0.000000) +(-45.000000:8.485281pt and 8.485281pt) -- +(45.000000:8.485281pt and 8.485281pt) -- +(135.000000:8.485281pt and 8.485281pt) -- +(225.000000:8.485281pt and 8.485281pt) -- cycle;
\clip (91.000000, -0.000000) +(-45.000000:8.485281pt and 8.485281pt) -- +(45.000000:8.485281pt and 8.485281pt) -- +(135.000000:8.485281pt and 8.485281pt) -- +(225.000000:8.485281pt and 8.485281pt) -- cycle;
\draw (91.000000, -0.000000) node {{$T$}};
\endscope
\draw (104.000000,28.000000) -- (104.000000,0.000000);
\filldraw (104.000000, 28.000000) circle(1.500000pt);
\scope
\draw[fill=white] (104.000000, 0.000000) circle(3.000000pt);
\clip (104.000000, 0.000000) circle(3.000000pt);
\draw (101.000000, 0.000000) -- (107.000000, 0.000000);
\draw (104.000000, -3.000000) -- (104.000000, 3.000000);
\endscope
\draw (114.000000,14.000000) -- (114.000000,0.000000);
\filldraw (114.000000, 14.000000) circle(1.500000pt);
\scope
\draw[fill=white] (114.000000, 0.000000) circle(3.000000pt);
\clip (114.000000, 0.000000) circle(3.000000pt);
\draw (111.000000, 0.000000) -- (117.000000, 0.000000);
\draw (114.000000, -3.000000) -- (114.000000, 3.000000);
\endscope
\draw (124.000000,28.000000) -- (124.000000,0.000000);
\filldraw (124.000000, 0.000000) circle(1.500000pt);
\scope
\draw[fill=white] (124.000000, 28.000000) circle(3.000000pt);
\clip (124.000000, 28.000000) circle(3.000000pt);
\draw (121.000000, 28.000000) -- (127.000000, 28.000000);
\draw (124.000000, 25.000000) -- (124.000000, 31.000000);
\endscope
\scope
\draw[fill=white] (124.000000, 14.000000) circle(3.000000pt);
\clip (124.000000, 14.000000) circle(3.000000pt);
\draw (121.000000, 14.000000) -- (127.000000, 14.000000);
\draw (124.000000, 11.000000) -- (124.000000, 17.000000);
\endscope
\scope
\draw[fill=white] (137.000000, 28.000000) +(-45.000000:8.485281pt and 8.485281pt) -- +(45.000000:8.485281pt and 8.485281pt) -- +(135.000000:8.485281pt and 8.485281pt) -- +(225.000000:8.485281pt and 8.485281pt) -- cycle;
\clip (137.000000, 28.000000) +(-45.000000:8.485281pt and 8.485281pt) -- +(45.000000:8.485281pt and 8.485281pt) -- +(135.000000:8.485281pt and 8.485281pt) -- +(225.000000:8.485281pt and 8.485281pt) -- cycle;
\draw (137.000000, 28.000000) node {{$T^\dagger$}};
\endscope
\scope
\draw[fill=white] (137.000000, 14.000000) +(-45.000000:8.485281pt and 8.485281pt) -- +(45.000000:8.485281pt and 8.485281pt) -- +(135.000000:8.485281pt and 8.485281pt) -- +(225.000000:8.485281pt and 8.485281pt) -- cycle;
\clip (137.000000, 14.000000) +(-45.000000:8.485281pt and 8.485281pt) -- +(45.000000:8.485281pt and 8.485281pt) -- +(135.000000:8.485281pt and 8.485281pt) -- +(225.000000:8.485281pt and 8.485281pt) -- cycle;
\draw (137.000000, 14.000000) node {{$T^\dagger$}};
\endscope
\scope
\draw[fill=white] (137.000000, -0.000000) +(-45.000000:8.485281pt and 8.485281pt) -- +(45.000000:8.485281pt and 8.485281pt) -- +(135.000000:8.485281pt and 8.485281pt) -- +(225.000000:8.485281pt and 8.485281pt) -- cycle;
\clip (137.000000, -0.000000) +(-45.000000:8.485281pt and 8.485281pt) -- +(45.000000:8.485281pt and 8.485281pt) -- +(135.000000:8.485281pt and 8.485281pt) -- +(225.000000:8.485281pt and 8.485281pt) -- cycle;
\draw (137.000000, -0.000000) node {{$T$}};
\endscope
\draw (150.000000,28.000000) -- (150.000000,0.000000);
\filldraw (150.000000, 0.000000) circle(1.500000pt);
\scope
\draw[fill=white] (150.000000, 28.000000) circle(3.000000pt);
\clip (150.000000, 28.000000) circle(3.000000pt);
\draw (147.000000, 28.000000) -- (153.000000, 28.000000);
\draw (150.000000, 25.000000) -- (150.000000, 31.000000);
\endscope
\scope
\draw[fill=white] (150.000000, 14.000000) circle(3.000000pt);
\clip (150.000000, 14.000000) circle(3.000000pt);
\draw (147.000000, 14.000000) -- (153.000000, 14.000000);
\draw (150.000000, 11.000000) -- (150.000000, 17.000000);
\endscope
\scope
\draw[fill=white] (163.000000, -0.000000) +(-45.000000:8.485281pt and 8.485281pt) -- +(45.000000:8.485281pt and 8.485281pt) -- +(135.000000:8.485281pt and 8.485281pt) -- +(225.000000:8.485281pt and 8.485281pt) -- cycle;
\clip (163.000000, -0.000000) +(-45.000000:8.485281pt and 8.485281pt) -- +(45.000000:8.485281pt and 8.485281pt) -- +(135.000000:8.485281pt and 8.485281pt) -- +(225.000000:8.485281pt and 8.485281pt) -- cycle;
\draw (163.000000, -0.000000) node {$H$};
\endscope
\scope
\draw[fill=white] (179.000000, -0.000000) +(-45.000000:8.485281pt and 8.485281pt) -- +(45.000000:8.485281pt and 8.485281pt) -- +(135.000000:8.485281pt and 8.485281pt) -- +(225.000000:8.485281pt and 8.485281pt) -- cycle;
\clip (179.000000, -0.000000) +(-45.000000:8.485281pt and 8.485281pt) -- +(45.000000:8.485281pt and 8.485281pt) -- +(135.000000:8.485281pt and 8.485281pt) -- +(225.000000:8.485281pt and 8.485281pt) -- cycle;
\draw (179.000000, -0.000000) node {{$S$}};
\endscope
\draw[color=black] (187.000000,28.000000) node[right] {$|x_1\rangle$};
\draw[color=black] (187.000000,14.000000) node[right] {$|x_2\rangle$};
\draw[color=black] (187.000000,0.000000) node[right] {$|x_1x_2\rangle$};
{\draw[color=black,fill=white] (12.000000,0.000000) node[right] {$|x_1x_2\rangle$};}
\endtikzpicture}

  \subfloat[$\cand^\dagger$ gate with no auxiliary qubits~\cite{Jones13,Gidney18}]{\tikzpicture[scale=1.000000,x=1pt,y=1pt]
\filldraw[color=white] (0.000000, -7.000000) rectangle (141.000000, 35.000000);
\footnotesize
\draw[color=black] (0.000000,28.000000) -- (19.500000,28.000000);
\draw[color=black] (57.500000,28.000000) -- (141.000000,28.000000);
\draw[color=black] (0.000000,28.000000) node[left] {$|x_1\rangle$};
\draw[color=black] (0.000000,14.000000) -- (19.500000,14.000000);
\draw[color=black] (57.500000,14.000000) -- (141.000000,14.000000);
\draw[color=black] (0.000000,14.000000) node[left] {$|x_2\rangle$};
\draw[color=black] (0.000000,0.000000) -- (5.000000,0.000000);
\draw[color=black,color=white] (5.000000,0.000000) -- (19.500000,0.000000);
\draw[color=black,color=black] (57.500000,0.000000) -- (107.000000,0.000000);
\draw[color=black,color=black] (107.000000,-0.500000) -- (141.000000,-0.500000);
\draw[color=black,color=black] (107.000000,0.500000) -- (141.000000,0.500000);
\draw[color=black] (0.000000,0.000000) node[left] {$|x_1x_2\rangle$};
\draw (5.000000,28.000000) -- (5.000000,0.000000);
\filldraw (5.000000, 28.000000) circle(1.500000pt);
\filldraw (5.000000, 14.000000) circle(1.500000pt);
\scope
\endscope
\scope
\endscope
\draw[color=black] (12.000000,28.000000) node[fill=white,right,minimum height=14.000000pt,minimum width=15.000000pt,inner sep=0pt] {\phantom{$|x_1\rangle$}};
\draw[color=black] (12.000000,28.000000) node[right] {$|x_1\rangle$};
\draw[color=black] (12.000000,14.000000) node[fill=white,right,minimum height=14.000000pt,minimum width=15.000000pt,inner sep=0pt] {\phantom{$|x_2\rangle$}};
\draw[color=black] (12.000000,14.000000) node[right] {$|x_2\rangle$};
\draw[color=black] (12.000000,0.000000) node[fill=white,right,minimum height=14.000000pt,minimum width=15.000000pt,inner sep=0pt] {\phantom{${}$}};
\draw[color=black] (12.000000,0.000000) node[right] {${}$};
\draw[fill=white,color=white] (31.000000, -6.000000) rectangle (46.000000, 34.000000);
\draw (38.500000, 14.000000) node {$=$};
\draw[color=black] (65.000000,28.000000) node[fill=white,left,minimum height=14.000000pt,minimum width=15.000000pt,inner sep=0pt] {\phantom{$|x_1\rangle$}};
\draw[color=black] (65.000000,28.000000) node[left] {$|x_1\rangle$};
\draw[color=black] (65.000000,14.000000) node[fill=white,left,minimum height=14.000000pt,minimum width=15.000000pt,inner sep=0pt] {\phantom{$|x_2\rangle$}};
\draw[color=black] (65.000000,14.000000) node[left] {$|x_2\rangle$};
\draw[color=black] (65.000000,0.000000) node[fill=white,left,minimum height=14.000000pt,minimum width=15.000000pt,inner sep=0pt] {\phantom{$|0\rangle$}};
\draw[color=black] (65.000000,0.000000) node[left] {$|0\rangle$};
\scope
\draw[fill=white] (75.000000, -0.000000) +(-45.000000:8.485281pt and 8.485281pt) -- +(45.000000:8.485281pt and 8.485281pt) -- +(135.000000:8.485281pt and 8.485281pt) -- +(225.000000:8.485281pt and 8.485281pt) -- cycle;
\clip (75.000000, -0.000000) +(-45.000000:8.485281pt and 8.485281pt) -- +(45.000000:8.485281pt and 8.485281pt) -- +(135.000000:8.485281pt and 8.485281pt) -- +(225.000000:8.485281pt and 8.485281pt) -- cycle;
\draw (75.000000, -0.000000) node {$H$};
\endscope
\scope
\draw[fill=white] (91.000000, 28.000000) +(-45.000000:8.485281pt and 8.485281pt) -- +(45.000000:8.485281pt and 8.485281pt) -- +(135.000000:8.485281pt and 8.485281pt) -- +(225.000000:8.485281pt and 8.485281pt) -- cycle;
\clip (91.000000, 28.000000) +(-45.000000:8.485281pt and 8.485281pt) -- +(45.000000:8.485281pt and 8.485281pt) -- +(135.000000:8.485281pt and 8.485281pt) -- +(225.000000:8.485281pt and 8.485281pt) -- cycle;
\draw (91.000000, 28.000000) node {{$S$}};
\endscope
\scope
\draw[fill=white] (91.000000, 14.000000) +(-45.000000:8.485281pt and 8.485281pt) -- +(45.000000:8.485281pt and 8.485281pt) -- +(135.000000:8.485281pt and 8.485281pt) -- +(225.000000:8.485281pt and 8.485281pt) -- cycle;
\clip (91.000000, 14.000000) +(-45.000000:8.485281pt and 8.485281pt) -- +(45.000000:8.485281pt and 8.485281pt) -- +(135.000000:8.485281pt and 8.485281pt) -- +(225.000000:8.485281pt and 8.485281pt) -- cycle;
\draw (91.000000, 14.000000) node {{$S$}};
\endscope
\draw (107.000000,28.000000) -- (107.000000,14.000000);
\filldraw (107.000000, 28.000000) circle(1.500000pt);
\scope
\draw[fill=white] (107.000000, 14.000000) circle(3.000000pt);
\clip (107.000000, 14.000000) circle(3.000000pt);
\draw (104.000000, 14.000000) -- (110.000000, 14.000000);
\draw (107.000000, 11.000000) -- (107.000000, 17.000000);
\endscope
\draw[fill=white] (101.000000, -6.000000) rectangle (113.000000, 6.000000);
\draw[very thin] (107.000000, 0.600000) arc (90:150:6.000000pt);
\draw[very thin] (107.000000, 0.600000) arc (90:30:6.000000pt);
\draw[->,>=stealth] (107.000000, -5.400000) -- +(80:10.392305pt);
\scope
\draw[fill=white] (123.000000, 14.000000) +(-45.000000:8.485281pt and 8.485281pt) -- +(45.000000:8.485281pt and 8.485281pt) -- +(135.000000:8.485281pt and 8.485281pt) -- +(225.000000:8.485281pt and 8.485281pt) -- cycle;
\clip (123.000000, 14.000000) +(-45.000000:8.485281pt and 8.485281pt) -- +(45.000000:8.485281pt and 8.485281pt) -- +(135.000000:8.485281pt and 8.485281pt) -- +(225.000000:8.485281pt and 8.485281pt) -- cycle;
\draw (123.000000, 14.000000) node {{$S^\dagger$}};
\endscope
\scope
\draw[fill=white] (123.000000, -0.000000) +(-45.000000:8.485281pt and 8.485281pt) -- +(45.000000:8.485281pt and 8.485281pt) -- +(135.000000:8.485281pt and 8.485281pt) -- +(225.000000:8.485281pt and 8.485281pt) -- cycle;
\clip (123.000000, -0.000000) +(-45.000000:8.485281pt and 8.485281pt) -- +(45.000000:8.485281pt and 8.485281pt) -- +(135.000000:8.485281pt and 8.485281pt) -- +(225.000000:8.485281pt and 8.485281pt) -- cycle;
\draw (123.000000, -0.000000) node {$X$};
\endscope
\draw (136.000000,28.000000) -- (136.000000,14.000000);
\filldraw (136.000000, 28.000000) circle(1.500000pt);
\scope
\draw[fill=white] (136.000000, 14.000000) circle(3.000000pt);
\clip (136.000000, 14.000000) circle(3.000000pt);
\draw (133.000000, 14.000000) -- (139.000000, 14.000000);
\draw (136.000000, 11.000000) -- (136.000000, 17.000000);
\endscope
\draw[color=black] (141.000000,28.000000) node[right] {$|x_1\rangle$};
\draw[color=black] (141.000000,14.000000) node[right] {$|x_2\rangle$};
\draw[color=black] (141.000000,0.000000) node[right] {$|x_1x_2\rangle$};
\draw[draw opacity=1.000000,fill opacity=0.200000,color=black,dashed] (84.000000,35.000000) rectangle (140.000000,7.000000);
\draw[draw opacity=1.000000,fill opacity=0.200000,color=black,dashed] (84.000000,35.000000) rectangle (140.000000,7.000000);
{\draw (106.5, 6) -- (106.5, 7);}
{\draw (107.5, 6) -- (107.5, 7);}
\endtikzpicture}
  \caption{Clifford+$T$ implementations for the \ccnot gate.}
  \label{fig:ccnot}
\end{figure}

We follow the usual convention of using quantum circuits as
computational models to describe quantum computations~\cite{NC00},
e.g., as already used in Fig.~\ref{fig:ccnot}.  We also use a textual
description for quantum circuits.  Assuming that qubits are indexed by
a distinct set of integers, we write $U_i$ to mean that a single qubit
unitary is applied to the qubit with index $i$, we write $\cnot_{i,j}$
when a \cnot gate is applied with control qubit $i$ and target qubit
$j$, and we use `$\circ$' for sequential composition of unitaries
(note that sequential composition is read from left to right, as in
the diagrammatic notation, and not from right to left as in matrix
multiplication).  We use $C^\dagger$ to denote the reverse circuit of
$C$, i.e., the reverse sequence where each unitary is replaced by its
adjoint.  Finally, we use the notation $[C]_i$ to mean that circuit
$C$ is only performed if the measurement outcome of qubit $i$ is $1$
when measured in the $Z$ basis. Note that this formalism does not
explicitly describe which computations are performed in parallel, but
we assume that gates are parallelized in a way such that the rotation
depth is minimized.  As an example, the circuit in
Fig.~\ref{fig:ccnot}(d) can be described as
$H_2 \circ [S_0 \circ S_1 \circ \cnot_{0,1} \circ
S^\dagger_1 \circ \cnot_{0,1} \circ X_2]_2$, assuming qubit indexes
$0, 1$, and $2$ for the qubits from top to bottom.

\paragraph*{Other notation used in the paper.} For a bitstring $x = x_1\dots
x_n$, let $\mu(x) = \sum_{i=1}^nx_i$ be the sideways sum of $x$, also called
Hamming weight.  Further, let $\rho x$ be the ruler function of $x$, which is
the largest integer $k$ such that $2^k \mathbin{|} x$, where $x \neq 0$.
Furthermore, we define $\rho 0 = \infty$. For two bitstrings $x$ and $y$ of the
same length, let $x \oplus y$ be the bitwise XOR operation.  When convenient, we
write $\mone$ in place of $-1$.

\section{General Case: Arbitrary Target Quantum State}

The proposed constructions make use of Gray codes and the Hadamard transform,
which we review in the beginning of this section.

\paragraph*{Gray codes.}  Let $v_0, v_1, \dots, v_{2^n-1}$ be a cyclic binary
Gray code that traverses all $n$-bit strings, i.e., $\mu(v_k \oplus v_{(k+1)
\bmod 2^n}) = 1$ and $v_k \neq v_l$ when $k \neq l$.  In other words, the
sequence forms a Hamiltonian cycle on the $n$-dimensional hypercube.  Without
loss of generality, we further assume that $v_0 = 0\dots 0$.  With this start
value, the sequence is uniquely determined by integers $\delta_0, \delta_1,
\dots, \delta_{2^n-1}$, such that $v_{(k+1)\bmod 2^n} = v_k \oplus
2^{\delta_k}$.  Note that for the standard Gray binary code, we have $\delta_k =
\rho(k+1)$ for $0 \le k < 2^n - 1$ and $\delta_{2^n-1} = n - 1$.

\begin{example}
  For $n = 2$, a cyclic binary Gray code is
  $v_0 = 00, v_1 = 01, v_2 = 11, v_3 = 10$, or alternatively expressed
  as $\delta_0 = 0, \delta_1 = 1, \delta_2 = 0, \delta_3 = 1$.
\end{example}

\paragraph*{Hadamard transform.}  The Hadamard transform $H_n$ on $n$
variables is a $2^n \times 2^n$ integral matrix which can be
recursively defined as $H_n = H_1 \otimes H_{n-1}$, where
$H_1 = \left(\begin{smallmatrix} 1&1 \\
    1&\mone \end{smallmatrix}\right)$.  For an $n$-variable Boolean
function $f(x_1, \dots, x_n)$ we define the $2^n$-element column
vector
\begin{multline}
  \label{eq:1-1coding}
  \hat f = ((-1)^{f(0, \dots, 0, 0)}, (-1)^{f(0, \dots, 0, 1)}, \dots, \\
  (-1)^{f(1, \dots, 1, 0)}, (-1)^{f(1, \dots, 1, 1)})^T,
\end{multline}
which contains all truth values of $f$, where $0$ and $1$ are encoded as $1$ and
$-1$, respectively.  The vector $\hat f$ is also called the (Walsh-)Hadamard
spectrum of $f$~\cite{MS77}. Multiplying $H_n \hat f$ results in a column vector
$s = (s_0, \dots, s_{2^n-1})^T$ whose elements are called the spectral
coefficients of $f$.

\begin{example}
  For $f(x_1, x_2) = x_1x_2$, we have $\hat f = (1, 1, 1, -1)^T$ and
  \[
    H_2\hat f =
    \begin{pmatrix}
      1 &     1 &     1 &     1 \\
      1 & \mone &     1 & \mone \\
      1 &     1 & \mone & \mone \\
      1 & \mone & \mone &     1
    \end{pmatrix}
    \begin{pmatrix}
      1 \\ 1 \\ 1 \\ \mone
    \end{pmatrix}
    =
    \begin{pmatrix}
      2 \\ 2 \\ 2 \\ -2
    \end{pmatrix}.
  \]
\end{example}

\begin{construction}[General case, no aux.~qubits]
  \label{con:ccnot-low-qubit}
  Let $f$ be an $n$-variable Boolean function and let
  $s = (s_0, \dots, s_{2^n-1})^T$ be its spectral coefficients.
  Assuming that the control qubit for variable $x_i$ has index $i-1$
  and the target qubit has index $n$, the circuit
  \begin{equation}
    \label{eq:ccnot-low-qubit}
    H_n \circ S_n \circ
    \mathop{\bigcirc}\limits_{i=0}^{n-1} C_i
    \circ C
    \circ
    H_n
  \end{equation}
  where
  \[
    C_i = \mathop{\bigcirc}\limits_{k=0}^{2^i-1}\left(
      R_1(\theta_{2^i+v_k})_i \circ \cnot_{\delta_k,i}\right)
  \]
  and
  \[
    C = \mathop{\bigcirc}\limits_{k=0}^{2^n-1}\left(
      R_1^\dagger(\theta_{v_k})_n \circ \cnot_{\delta_k,n} \right)
  \]
  implements
  $U_f:|x\rangle|y\rangle \mapsto |x\rangle|y\oplus f(x)\rangle$
  without any auxiliary qubits, where
  $\theta_j = \tfrac{s_j\pi}{2^{n+1}}$ and $\cnot_{-1,0} = I_0$ (this
  case occurs once in subcircuit $C_0$ when $k=0$.)
\end{construction}
\begin{figure*}[t]
  \centering
  \input{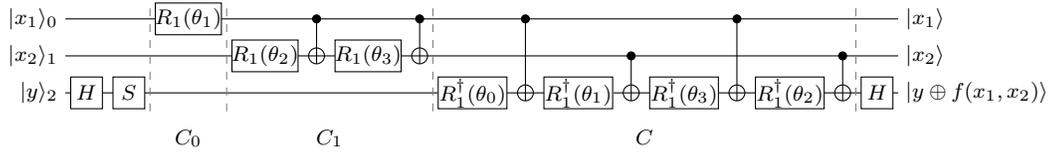}
  \caption{Example circuit for Construction~\ref{con:ccnot-low-qubit}
    where $n=2$.  The subscripts on the input qubit label indicate the
    qubit index.}
  \label{fig:ccnot-low-qubit}
\end{figure*}
For each value $1 \le k < 2^n$, with binary expansion $k = (b_nb_{n-1}\dots
b_1)_2$, we apply the gate $R_1(\theta_k)$. This phase must be applied to a
computational state corresponding to $x_1^{b_1} \oplus \cdots \oplus x_n^{b_n}$,
which can be constructed using \cnot gates~\cite{DHH+05,SS03,AMM14,AM19}.  In
order to reduce the number of required \cnot gates, the subcircuit $C_i$ applies
them using a cyclic binary Gray code~\cite{WGMA14}, for all values $k$ in which
the leading 1 is at position $b_{i+1}$, i.e., $2^i \le k < 2^{i+1}$.  Since the
Gray code is cyclic, the computational state of the qubits remains unchanged and
corresponds to the input qubits after each application of subcircuit $C_i$.  The
last subcircuit $C$ prepares all computational states corresponding to all
linear combinations of the input qubits together with the target qubit.  In this
case, the adjoint of the corresponding rotation gate is applied.  Why these
rotation angles correspond to the spectral coefficients of $f$ is explained and
proven in Appendix~\ref{sec:app-rotations}.
\begin{example}
  Fig.~\ref{fig:ccnot-low-qubit} shows the circuit from
  Construction~\ref{con:ccnot-low-qubit} when $n = 2$.
\end{example}
Note that the $S$ gate may be merged with the $R_1^\dagger(\theta_0)$
gate as $R_1^\dagger(\theta_0 + \tfrac{\pi}{2})$.  When $f = x_1x_2$
and when moving all gates as far as left as possible, one obtains the
circuit in Fig.~\ref{fig:ccnot}(a).

\begin{construction}[General case, depth 1]
  \label{con:ccnot-depth-1}
  Let $f(x_1, \dots, x_n)$ be a Boolean function and let
  $s = (s_0, s_1, \dots, s_{2^n-1})^T$ be its spectral coefficients.
  The circuit we construct has $2^{n+1} - 1$ qubits indexed from $1$
  to $2^{n+1} - 1$.  Assuming that the control qubit for variable
  $x_i$ has index $2^{i-1}$, that the target qubit has index $2^n$,
  and that all other indexes are assigned to the auxiliary qubits, the
  circuit
  \begin{equation}
    \label{eq:ccnot-depth-1}
    H_{2^n} \circ S_{2^n} \circ
    C_1 \circ C_2 \circ
    R \circ
    C_2^\dagger \circ C_1^\dagger
    \circ H_{2^n},
  \end{equation}
  where
  \[
    C_1 =
    \mathop{\bigcirc}\limits_{\substack{3 \le k < 2^{n+1} \\ \mu k \neq 1}}
      \cnot_{\rho k,k},\;
    C_2 =
    \mathop{\bigcirc}\limits_{\substack{3 \le k < 2^{n+1} \\ \mu k \neq 1}}
      \cnot_{k - \rho k,k}
  \]
  and
  \[
    R = \mathop{\bigcirc}\limits_{k=1}^{2^n-1} R_1(\theta_k)_k \circ
    \mathop{\bigcirc}\limits_{k=0}^{2^n-1} R_1^\dagger(\theta_k)_{2^n+k}
  \]
  implements
  $U_f : |x\rangle|y\rangle|0^\ell\rangle \mapsto |x\rangle|y\oplus
  f(x)\rangle|0^\ell\rangle$, where $\ell = 2^{n+1}-n-2$ and
  $\theta_k = \tfrac{s_k\pi}{2^{n+1}}$.  Note that all rotation gates
  in~$R$ can be executed in parallel, since they are all applied to
  different qubits.
\end{construction}
The auxiliary qubit with index $k = (b_{n+1}\dots b_1)_2$ is used to
prepare the computational state of the linear combination
$x_1^{b_1} \oplus \cdots \oplus x_n^{b_n} \oplus y^{b_{n+1}}$.
Subcircuit $C_1$ initializes the auxiliary qubits with $\rho k$, i.e.,
the trailing 1 in $k$ and subcircuit $C_2$ copies over the remaining
bits using $k - \rho k$.  The order guarantees that this is always
possible.
\begin{figure}[t]
  \centering
  \input{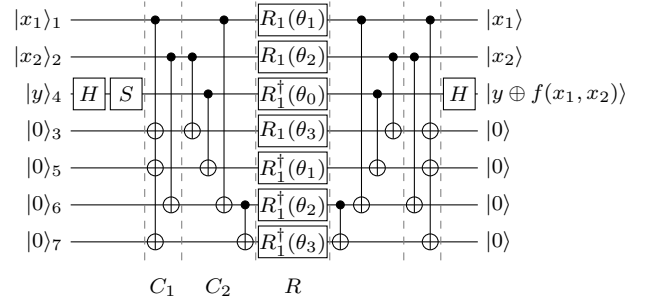}
  \caption{Example circuit for Construction~\ref{con:ccnot-depth-1}
    where $n=2$.  The subscripts on the input qubit label indicate the
    qubit index.}
  \label{fig:ccnot-depth-1}
\end{figure}
\begin{example}
  Fig.~\ref{fig:ccnot-depth-1} shows the circuit from
  Construction~\ref{con:ccnot-depth-1} when $n = 2$.
\end{example}
As in Construction~\ref{con:ccnot-low-qubit}, the $S$ gate may be
merged with the $R_1^\dagger(\theta_0)$ gate as
$R_1^\dagger(\theta_0 + \tfrac{\pi}{2})$.  When $f = x_1x_2$ and after
merging the $S$ gate, one obtains the circuit in
Fig.~\ref{fig:ccnot}(b).

\section{Special Case: Target Quantum State is $|0\rangle$ or $|f(x)\rangle$}
In this section, we show constructions for generalizations of $\cand$
and $\cand^\dagger$, i.e., unitary operations in which the target
quantum state is known to be $|0\rangle$ or $|f(x)\rangle$,
respectively.
\begin{construction}[Target $|0\rangle$, no aux.~qubits]
  \label{con:and-low-qubit}
  Let $f$ be an $n$-variable Boolean function and let
  $s = (s_0, \dots, s_{2^n-1})^T$ be its spectral coefficients.  The
  circuit

  \begin{equation}
    \label{eq:and-low-qubit}
    H_n \circ S_n \circ
    \mathop{\bigcirc}\limits_{k=0}^{2^n-1}\left(
      R_1^\dagger(\theta_{v_k})_n \circ \cnot_{\delta_k,n}
    \right)
    \circ
    H_n
  \end{equation}
  implements $U_f:|x\rangle|0\rangle \mapsto |x\rangle|f(x)\rangle$
  without any auxiliary qubits, where
  $\theta_j = \tfrac{s_j\pi}{2^{n+1}}$.  Note that this equals
  $H_n \circ S_n \circ C \circ H_n$, where $C$ is as
  in~\eqref{eq:ccnot-low-qubit} from
  Construction~\ref{con:ccnot-low-qubit}.
\end{construction}
If the target qubit is in state $|0\rangle$, we only need to apply
$R_1$ gates to computational states that involve the target qubit.
Therefore, almost half of the rotation gates can be saved, in fact all
subcircuits $C_i$, which are required in
Construction~\ref{con:ccnot-low-qubit}, can be omitted if the target
qubit is known to be in state $|0\rangle$.
\begin{figure}[t]
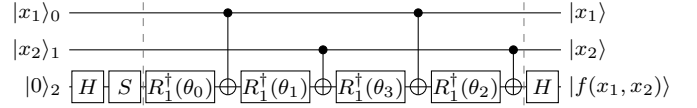

  \centering
  \tikzpicture[scale=1.000000,x=1pt,y=1pt]
\filldraw[color=white] (0.000000, -7.000000) rectangle (186.000000, 35.000000);
\footnotesize
\draw[color=black] (0.000000,28.000000) -- (186.000000,28.000000);
\draw[color=black] (0.000000,28.000000) node[left] {$|x_1\rangle_0$};
\draw[color=black] (0.000000,14.000000) -- (186.000000,14.000000);
\draw[color=black] (0.000000,14.000000) node[left] {$|x_2\rangle_1$};
\draw[color=black] (0.000000,0.000000) -- (186.000000,0.000000);
\draw[color=black] (0.000000,0.000000) node[left] {$|0\rangle_2$};
\scope
\draw[fill=white] (7.000000, -0.000000) +(-45.000000:8.485281pt and 8.485281pt) -- +(45.000000:8.485281pt and 8.485281pt) -- +(135.000000:8.485281pt and 8.485281pt) -- +(225.000000:8.485281pt and 8.485281pt) -- cycle;
\clip (7.000000, -0.000000) +(-45.000000:8.485281pt and 8.485281pt) -- +(45.000000:8.485281pt and 8.485281pt) -- +(135.000000:8.485281pt and 8.485281pt) -- +(225.000000:8.485281pt and 8.485281pt) -- cycle;
\draw (7.000000, -0.000000) node {$H$};
\endscope
\scope
\draw[fill=white] (21.000000, -0.000000) +(-45.000000:8.485281pt and 8.485281pt) -- +(45.000000:8.485281pt and 8.485281pt) -- +(135.000000:8.485281pt and 8.485281pt) -- +(225.000000:8.485281pt and 8.485281pt) -- cycle;
\clip (21.000000, -0.000000) +(-45.000000:8.485281pt and 8.485281pt) -- +(45.000000:8.485281pt and 8.485281pt) -- +(135.000000:8.485281pt and 8.485281pt) -- +(225.000000:8.485281pt and 8.485281pt) -- cycle;
\draw (21.000000, -0.000000) node {{$S$}};
\endscope
\scope[color=gray,dashed]
\draw (28.000000,-7.000000) -- (28.000000,35.000000);
\endscope
\scope
\draw[fill=white] (42.000000, -0.000000) +(-45.000000:18.384776pt and 8.485281pt) -- +(45.000000:18.384776pt and 8.485281pt) -- +(135.000000:18.384776pt and 8.485281pt) -- +(225.000000:18.384776pt and 8.485281pt) -- cycle;
\clip (42.000000, -0.000000) +(-45.000000:18.384776pt and 8.485281pt) -- +(45.000000:18.384776pt and 8.485281pt) -- +(135.000000:18.384776pt and 8.485281pt) -- +(225.000000:18.384776pt and 8.485281pt) -- cycle;
\draw (42.000000, -0.000000) node {{$R_1^\dagger(\theta_0)$}};
\endscope
\draw (60.000000,28.000000) -- (60.000000,0.000000);
\filldraw (60.000000, 28.000000) circle(1.500000pt);
\scope
\draw[fill=white] (60.000000, 0.000000) circle(3.000000pt);
\clip (60.000000, 0.000000) circle(3.000000pt);
\draw (57.000000, 0.000000) -- (63.000000, 0.000000);
\draw (60.000000, -3.000000) -- (60.000000, 3.000000);
\endscope
\scope
\draw[fill=white] (78.000000, -0.000000) +(-45.000000:18.384776pt and 8.485281pt) -- +(45.000000:18.384776pt and 8.485281pt) -- +(135.000000:18.384776pt and 8.485281pt) -- +(225.000000:18.384776pt and 8.485281pt) -- cycle;
\clip (78.000000, -0.000000) +(-45.000000:18.384776pt and 8.485281pt) -- +(45.000000:18.384776pt and 8.485281pt) -- +(135.000000:18.384776pt and 8.485281pt) -- +(225.000000:18.384776pt and 8.485281pt) -- cycle;
\draw (78.000000, -0.000000) node {{$R_1^\dagger(\theta_1)$}};
\endscope
\draw (96.000000,14.000000) -- (96.000000,0.000000);
\filldraw (96.000000, 14.000000) circle(1.500000pt);
\scope
\draw[fill=white] (96.000000, 0.000000) circle(3.000000pt);
\clip (96.000000, 0.000000) circle(3.000000pt);
\draw (93.000000, 0.000000) -- (99.000000, 0.000000);
\draw (96.000000, -3.000000) -- (96.000000, 3.000000);
\endscope
\scope
\draw[fill=white] (114.000000, -0.000000) +(-45.000000:18.384776pt and 8.485281pt) -- +(45.000000:18.384776pt and 8.485281pt) -- +(135.000000:18.384776pt and 8.485281pt) -- +(225.000000:18.384776pt and 8.485281pt) -- cycle;
\clip (114.000000, -0.000000) +(-45.000000:18.384776pt and 8.485281pt) -- +(45.000000:18.384776pt and 8.485281pt) -- +(135.000000:18.384776pt and 8.485281pt) -- +(225.000000:18.384776pt and 8.485281pt) -- cycle;
\draw (114.000000, -0.000000) node {{$R_1^\dagger(\theta_3)$}};
\endscope
\draw (132.000000,28.000000) -- (132.000000,0.000000);
\filldraw (132.000000, 28.000000) circle(1.500000pt);
\scope
\draw[fill=white] (132.000000, 0.000000) circle(3.000000pt);
\clip (132.000000, 0.000000) circle(3.000000pt);
\draw (129.000000, 0.000000) -- (135.000000, 0.000000);
\draw (132.000000, -3.000000) -- (132.000000, 3.000000);
\endscope
\scope
\draw[fill=white] (150.000000, -0.000000) +(-45.000000:18.384776pt and 8.485281pt) -- +(45.000000:18.384776pt and 8.485281pt) -- +(135.000000:18.384776pt and 8.485281pt) -- +(225.000000:18.384776pt and 8.485281pt) -- cycle;
\clip (150.000000, -0.000000) +(-45.000000:18.384776pt and 8.485281pt) -- +(45.000000:18.384776pt and 8.485281pt) -- +(135.000000:18.384776pt and 8.485281pt) -- +(225.000000:18.384776pt and 8.485281pt) -- cycle;
\draw (150.000000, -0.000000) node {{$R_1^\dagger(\theta_2)$}};
\endscope
\draw (168.000000,14.000000) -- (168.000000,0.000000);
\filldraw (168.000000, 14.000000) circle(1.500000pt);
\scope
\draw[fill=white] (168.000000, 0.000000) circle(3.000000pt);
\clip (168.000000, 0.000000) circle(3.000000pt);
\draw (165.000000, 0.000000) -- (171.000000, 0.000000);
\draw (168.000000, -3.000000) -- (168.000000, 3.000000);
\endscope
\scope[color=gray,dashed]
\draw (172.000000,-7.000000) -- (172.000000,35.000000);
\endscope
\scope
\draw[fill=white] (179.000000, -0.000000) +(-45.000000:8.485281pt and 8.485281pt) -- +(45.000000:8.485281pt and 8.485281pt) -- +(135.000000:8.485281pt and 8.485281pt) -- +(225.000000:8.485281pt and 8.485281pt) -- cycle;
\clip (179.000000, -0.000000) +(-45.000000:8.485281pt and 8.485281pt) -- +(45.000000:8.485281pt and 8.485281pt) -- +(135.000000:8.485281pt and 8.485281pt) -- +(225.000000:8.485281pt and 8.485281pt) -- cycle;
\draw (179.000000, -0.000000) node {$H$};
\endscope
\draw[color=black] (186.000000,28.000000) node[right] {$|x_1\rangle$};
\draw[color=black] (186.000000,14.000000) node[right] {$|x_2\rangle$};
\draw[color=black] (186.000000,0.000000) node[right] {${|f(x_1,x_2)\rangle}$};
\endtikzpicture
  \caption{Example circuit for Construction~\ref{con:and-low-qubit}
    where $n=2$.  The subscripts on the input qubit label indicate the
    qubit index.}
  \label{fig:and-low-qubit}
\end{figure}
\begin{example}
  Fig.~\ref{fig:and-low-qubit} shows the circuit from
  Construction~\ref{con:and-low-qubit} when $n = 2$.
\end{example}
Note that the circuit in Fig.~\ref{con:and-low-qubit} is quite
different from the circuit implementing \cand in
Fig.~\ref{fig:ccnot}(c).  This is because the circuit in
Fig.~\ref{fig:ccnot}(c) aims at reducing the rotation depth, which is
optimal in that case when no auxiliary qubits are used.  In order to
achieve the better rotation depth, more \cnot gates are required.
Similarly, a \ccnot gate with rotation depth 3 is possible, when using
more \cnot gates as in Fig.~\ref{fig:ccnot}(a)~\cite{AMMR13}.
However, requiring to distribute $2^n$ rotation gates over $n + 1$
qubits asymptotically still leads to an exponential worst case
complexity for the rotation depth.
\begin{construction}[Target $|0\rangle$, depth 1]
  Let $f$ be an $n$-variable Boolean function with spectral
  coefficients $(s_0, \dots, s_{2^n-1})^T$.  Assuming that control
  qubit $x_i$ is indexed $2^{i-1}$ and that the target qubit has index
  $0$, the circuit
  \label{con:and-depth-1}
  \begin{equation}
    \label{eq:and-depth-1}
    H_0 \circ S_0 \circ C_1 \circ C_3 \circ C_2 \circ R \circ C_2^\dagger \circ C_3^\dagger \circ C_1^\dagger \circ H_0,
  \end{equation}
  where
  \[
    C_1 =
      \mathop{\bigcirc}\limits_{\substack{3 \le k < 2^n \\ \mu k \neq 1}}
      \cnot_{\rho k,k}, \;
    C_3 = \mathop{\bigcirc}\limits_{i=0}^{n-1}\cnot_{0,2^i},
  \]
  \[
    C_2 =
      \mathop{\bigcirc}\limits_{\substack{3 \le k < 2^n \\ \mu k \neq 1}}
      \cnot_{k - \rho k,k},\;
    R = \mathop{\bigcirc}_{k=0}^{2^n-1}R_1^\dagger(\theta_k)_k
  \]
  implements
  $U_f : |x\rangle|0\rangle|0^\ell\rangle \mapsto
  |x\rangle|f(x)\rangle|0^\ell\rangle$, where $\ell = 2^n - n - 1$ and
  $\theta_j = \frac{s_j\pi}{2^{n+1}}$.
\end{construction}
One possible way to get a circuit with rotation depth 1 is to take the
circuit from Construction~\ref{con:ccnot-depth-1} and remove all $R_1$
gates that are on qubits with indexes smaller than $2^n$.  However,
this requires unnecessarily many auxiliary qubits, since only roughly
half as many rotations are needed as in
Construction~\ref{con:ccnot-depth-1}.  Subcircuits $C_1$ and $C_2$
play the same role as in Construction~\ref{con:ccnot-depth-1}, and
subcircuit $C_3$ ensures to add the target qubit to all the linear
combinations of input qubits in the computational states.
\begin{figure}[t]
  \centering
  \input{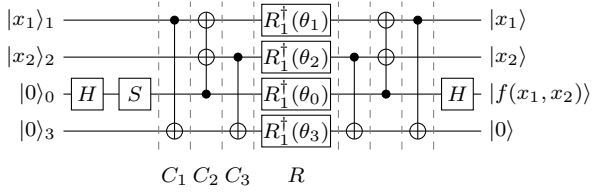}
  \caption{Example circuit for Construction~\ref{con:and-depth-1}
    where $n=2$.  The subscripts on the input qubit label indicate the
    qubit index.}
  \label{fig:and-depth-1}
\end{figure}
\begin{example}
  Fig.~\ref{fig:and-depth-1} shows the circuit from
  Construction~\ref{con:and-depth-1} when $n = 2$.
\end{example}

\begin{construction}[Target $|f(x)\rangle$, no aux.~qubits]
  Let $f$ be an $n$-variable Boolean function with spectral
  coefficients $(s_0, \dots, s_{2^n-1})^T$.  Assuming qubit index
  $i - 1$ for variable $x_i$ and qubit index $n$ for the target qubit,
  the circuit
  \label{con:anddg-low-qubit}
  \begin{equation}
    \label{eq:anddg-low-qubit}
    H_2 \circ \left[
      \mathop{\bigcirc}_{i=0}^{n-1} C_i
      \circ X_n
    \right]_n,
  \end{equation}
  where
  \[
    C_i = \mathop{\bigcirc}\limits_{k=0}^{2^i-1}\left(
      R_1(2\theta_{2^i+v_k})_i \circ \cnot_{\delta_k,i}\right).
  \]
  implements $U_f : |x\rangle|f(x)\rangle \mapsto |x\rangle|0\rangle$
  without any auxiliary qubits, where
  $\theta_j = \frac{s_j\pi}{2^{n+1}}$ and $\cnot_{-1,0} = I$ (this
  case occurs once in subcircuit $C_0$ when $k = 0$.)
\end{construction}
While Construction~\ref{con:and-low-qubit} only requires subcircuit
$C$ from Construction~\ref{con:ccnot-low-qubit}, in this case
subcircuit $C$ can be omitted, but the subcircuits $C_i$ are required
to perform a phase correction in case of a positive measurement
outcome.  Note, however, that the angle is doubled in subcircuit $C_i$
compared to Construction~\ref{con:ccnot-low-qubit}.  A larger rotation
angle typically affects the resource costs of an error-corrected
rotation gate positively in a fault-tolerant quantum computer.
\begin{figure}[t]
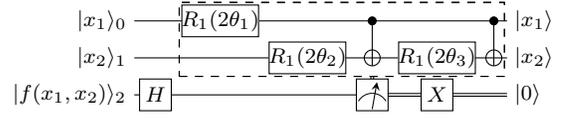

  \centering
  \tikzpicture[scale=1.000000,x=1pt,y=1pt]
\filldraw[color=white] (0.000000, -7.000000) rectangle (141.000000, 35.000000);
\footnotesize
\draw[color=black] (0.000000,28.000000) -- (141.000000,28.000000);
\draw[color=black] (0.000000,28.000000) node[left] {$|x_1\rangle_0$};
\draw[color=black] (0.000000,14.000000) -- (141.000000,14.000000);
\draw[color=black] (0.000000,14.000000) node[left] {$|x_2\rangle_1$};
\draw[color=black] (0.000000,0.000000) -- (90.000000,0.000000);
\draw[color=black] (90.000000,-0.500000) -- (141.000000,-0.500000);
\draw[color=black] (90.000000,0.500000) -- (141.000000,0.500000);
\draw[color=black] (0.000000,0.000000) node[left] {$|f(x_1,x_2)\rangle_2$};
\scope
\draw[fill=white] (8.000000, -0.000000) +(-45.000000:8.485281pt and 8.485281pt) -- +(45.000000:8.485281pt and 8.485281pt) -- +(135.000000:8.485281pt and 8.485281pt) -- +(225.000000:8.485281pt and 8.485281pt) -- cycle;
\clip (8.000000, -0.000000) +(-45.000000:8.485281pt and 8.485281pt) -- +(45.000000:8.485281pt and 8.485281pt) -- +(135.000000:8.485281pt and 8.485281pt) -- +(225.000000:8.485281pt and 8.485281pt) -- cycle;
\draw (8.000000, -0.000000) node {$H$};
\endscope
\scope
\draw[fill=white] (32.500000, 28.000000) +(-45.000000:20.506097pt and 8.485281pt) -- +(45.000000:20.506097pt and 8.485281pt) -- +(135.000000:20.506097pt and 8.485281pt) -- +(225.000000:20.506097pt and 8.485281pt) -- cycle;
\clip (32.500000, 28.000000) +(-45.000000:20.506097pt and 8.485281pt) -- +(45.000000:20.506097pt and 8.485281pt) -- +(135.000000:20.506097pt and 8.485281pt) -- +(225.000000:20.506097pt and 8.485281pt) -- cycle;
\draw (32.500000, 28.000000) node {{$R_1(2\theta_1)$}};
\endscope
\scope
\draw[fill=white] (65.500000, 14.000000) +(-45.000000:20.506097pt and 8.485281pt) -- +(45.000000:20.506097pt and 8.485281pt) -- +(135.000000:20.506097pt and 8.485281pt) -- +(225.000000:20.506097pt and 8.485281pt) -- cycle;
\clip (65.500000, 14.000000) +(-45.000000:20.506097pt and 8.485281pt) -- +(45.000000:20.506097pt and 8.485281pt) -- +(135.000000:20.506097pt and 8.485281pt) -- +(225.000000:20.506097pt and 8.485281pt) -- cycle;
\draw (65.500000, 14.000000) node {{$R_1(2\theta_2)$}};
\endscope
\draw[fill=white] (84.000000, -6.000000) rectangle (96.000000, 6.000000);
\draw[very thin] (90.000000, 0.600000) arc (90:150:6.000000pt);
\draw[very thin] (90.000000, 0.600000) arc (90:30:6.000000pt);
\draw[->,>=stealth] (90.000000, -5.400000) -- +(80:10.392305pt);
\draw (90.000000,28.000000) -- (90.000000,14.000000);
\filldraw (90.000000, 28.000000) circle(1.500000pt);
\scope
\draw[fill=white] (90.000000, 14.000000) circle(3.000000pt);
\clip (90.000000, 14.000000) circle(3.000000pt);
\draw (87.000000, 14.000000) -- (93.000000, 14.000000);
\draw (90.000000, 11.000000) -- (90.000000, 17.000000);
\endscope
\scope
\draw[fill=white] (114.500000, 14.000000) +(-45.000000:20.506097pt and 8.485281pt) -- +(45.000000:20.506097pt and 8.485281pt) -- +(135.000000:20.506097pt and 8.485281pt) -- +(225.000000:20.506097pt and 8.485281pt) -- cycle;
\clip (114.500000, 14.000000) +(-45.000000:20.506097pt and 8.485281pt) -- +(45.000000:20.506097pt and 8.485281pt) -- +(135.000000:20.506097pt and 8.485281pt) -- +(225.000000:20.506097pt and 8.485281pt) -- cycle;
\draw (114.500000, 14.000000) node {{$R_1(2\theta_3)$}};
\endscope
\scope
\draw[fill=white] (114.500000, -0.000000) +(-45.000000:8.485281pt and 8.485281pt) -- +(45.000000:8.485281pt and 8.485281pt) -- +(135.000000:8.485281pt and 8.485281pt) -- +(225.000000:8.485281pt and 8.485281pt) -- cycle;
\clip (114.500000, -0.000000) +(-45.000000:8.485281pt and 8.485281pt) -- +(45.000000:8.485281pt and 8.485281pt) -- +(135.000000:8.485281pt and 8.485281pt) -- +(225.000000:8.485281pt and 8.485281pt) -- cycle;
\draw (114.500000, -0.000000) node {$X$};
\endscope
\draw (136.000000,28.000000) -- (136.000000,14.000000);
\filldraw (136.000000, 28.000000) circle(1.500000pt);
\scope
\draw[fill=white] (136.000000, 14.000000) circle(3.000000pt);
\clip (136.000000, 14.000000) circle(3.000000pt);
\draw (133.000000, 14.000000) -- (139.000000, 14.000000);
\draw (136.000000, 11.000000) -- (136.000000, 17.000000);
\endscope
\draw[color=black] (141.000000,28.000000) node[right] {$|x_1\rangle$};
\draw[color=black] (141.000000,14.000000) node[right] {$|x_2\rangle$};
\draw[color=black] (141.000000,0.000000) node[right] {${|0\rangle}$};
\draw[draw opacity=1.000000,fill opacity=0.200000,color=black,dashed] (17.000000,35.000000) rectangle (140.000000,7.000000);
\draw[draw opacity=1.000000,fill opacity=0.200000,color=black,dashed] (17.000000,35.000000) rectangle (140.000000,7.000000);
{\draw (89.5, 6) -- (89.5, 7);}
{\draw (90.5, 6) -- (90.5, 7);}
\endtikzpicture
  \caption{Example circuit for Construction~\ref{con:anddg-low-qubit}
    where $n=2$.  The subscripts on the input qubit label indicate the
    qubit index.}
  \label{fig:anddg-low-qubit}
\end{figure}
\begin{example}
  Fig.~\ref{fig:anddg-low-qubit} shows the circuit from
  Construction~\ref{con:anddg-low-qubit} when $n = 2$.
\end{example}
When $f = x_1x_2$ and when moving all gates as far to the left as
possible, one obtains the circuit in Fig.~\ref{fig:ccnot}(d).

\begin{construction}[Target $|f(x)\rangle$, depth 1]
  Let $f$ be an $n$-variable Boolean function with spectral
  coefficients $(s_0, \dots, s_{2^n-1})^T$.  Assuming that control
  qubit $x_i$ is indexed $2^i$ and that the target qubit has index
  $0$, the circuit
  \label{con:anddg-depth-1}
  \begin{equation}
    \label{eq:anddg-depth-1}
    H_0 \circ
    \left[
      C_1 \circ C_2 \circ \mathop{\bigcirc}_{k=1}^{2^n-1}R_1(2\theta_k)_k \circ C_2^\dagger \circ C_1 \circ
      X_0
    \right],
  \end{equation}
  where
  \[
    C_1 =
    \mathop{\bigcirc}\limits_{\substack{3 \le k < 2^n \\ \mu k \neq 1}}
    \cnot_{\rho k,k}, \;
    C_2 =
    \mathop{\bigcirc}\limits_{\substack{3 \le k < 2^n \\ \mu k \neq 1}}
    \cnot_{k - \rho k,k}.
  \]
  implements
  $U_f : |x\rangle|f(x)\rangle|0^\ell\rangle \mapsto
  |x\rangle|0\rangle|0^\ell\rangle$, where $\ell = 2^n - n - 1$ and
  $\theta_j = \frac{s_j\pi}{2^{n+1}}$.
\end{construction}
The subcircuits $C_1$ and $C_2$ play the same role as in the general
case in Construction~\ref{con:ccnot-depth-1}, however, only
computational states of linear combinations involving the control
lines need to be prepared in order to apply the conditional phase
correction in case of a positive measurement result.  As in
Construction~\ref{con:anddg-low-qubit} the rotation angles are
doubled.
\begin{figure}[t]
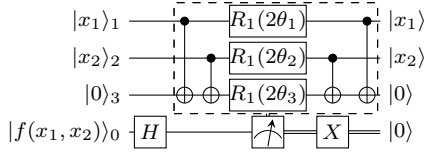

  \centering
  \tikzpicture[scale=1.000000,x=1pt,y=1pt]
\filldraw[color=white] (0.000000, -7.000000) rectangle (95.000000, 49.000000);
\footnotesize
\draw[color=black] (0.000000,42.000000) -- (95.000000,42.000000);
\draw[color=black] (0.000000,42.000000) node[left] {$|x_1\rangle_1$};
\draw[color=black] (0.000000,28.000000) -- (95.000000,28.000000);
\draw[color=black] (0.000000,28.000000) node[left] {$|x_2\rangle_2$};
\draw[color=black] (0.000000,14.000000) -- (95.000000,14.000000);
\draw[color=black] (0.000000,14.000000) node[left] {$|0\rangle_3$};
\draw[color=black] (0.000000,0.000000) -- (52.500000,0.000000);
\draw[color=black] (52.500000,-0.500000) -- (95.000000,-0.500000);
\draw[color=black] (52.500000,0.500000) -- (95.000000,0.500000);
\draw[color=black] (0.000000,0.000000) node[left] {$|f(x_1,x_2)\rangle_0$};
\scope
\draw[fill=white] (8.000000, -0.000000) +(-45.000000:8.485281pt and 8.485281pt) -- +(45.000000:8.485281pt and 8.485281pt) -- +(135.000000:8.485281pt and 8.485281pt) -- +(225.000000:8.485281pt and 8.485281pt) -- cycle;
\clip (8.000000, -0.000000) +(-45.000000:8.485281pt and 8.485281pt) -- +(45.000000:8.485281pt and 8.485281pt) -- +(135.000000:8.485281pt and 8.485281pt) -- +(225.000000:8.485281pt and 8.485281pt) -- cycle;
\draw (8.000000, -0.000000) node {$H$};
\endscope
\draw (21.000000,42.000000) -- (21.000000,14.000000);
\filldraw (21.000000, 42.000000) circle(1.500000pt);
\scope
\draw[fill=white] (21.000000, 14.000000) circle(3.000000pt);
\clip (21.000000, 14.000000) circle(3.000000pt);
\draw (18.000000, 14.000000) -- (24.000000, 14.000000);
\draw (21.000000, 11.000000) -- (21.000000, 17.000000);
\endscope
\draw (31.000000,28.000000) -- (31.000000,14.000000);
\filldraw (31.000000, 28.000000) circle(1.500000pt);
\scope
\draw[fill=white] (31.000000, 14.000000) circle(3.000000pt);
\clip (31.000000, 14.000000) circle(3.000000pt);
\draw (28.000000, 14.000000) -- (34.000000, 14.000000);
\draw (31.000000, 11.000000) -- (31.000000, 17.000000);
\endscope
\scope
\draw[fill=white] (52.500000, 42.000000) +(-45.000000:20.506097pt and 8.485281pt) -- +(45.000000:20.506097pt and 8.485281pt) -- +(135.000000:20.506097pt and 8.485281pt) -- +(225.000000:20.506097pt and 8.485281pt) -- cycle;
\clip (52.500000, 42.000000) +(-45.000000:20.506097pt and 8.485281pt) -- +(45.000000:20.506097pt and 8.485281pt) -- +(135.000000:20.506097pt and 8.485281pt) -- +(225.000000:20.506097pt and 8.485281pt) -- cycle;
\draw (52.500000, 42.000000) node {{$R_1(2\theta_1)$}};
\endscope
\scope
\draw[fill=white] (52.500000, 28.000000) +(-45.000000:20.506097pt and 8.485281pt) -- +(45.000000:20.506097pt and 8.485281pt) -- +(135.000000:20.506097pt and 8.485281pt) -- +(225.000000:20.506097pt and 8.485281pt) -- cycle;
\clip (52.500000, 28.000000) +(-45.000000:20.506097pt and 8.485281pt) -- +(45.000000:20.506097pt and 8.485281pt) -- +(135.000000:20.506097pt and 8.485281pt) -- +(225.000000:20.506097pt and 8.485281pt) -- cycle;
\draw (52.500000, 28.000000) node {{$R_1(2\theta_2)$}};
\endscope
\scope
\draw[fill=white] (52.500000, 14.000000) +(-45.000000:20.506097pt and 8.485281pt) -- +(45.000000:20.506097pt and 8.485281pt) -- +(135.000000:20.506097pt and 8.485281pt) -- +(225.000000:20.506097pt and 8.485281pt) -- cycle;
\clip (52.500000, 14.000000) +(-45.000000:20.506097pt and 8.485281pt) -- +(45.000000:20.506097pt and 8.485281pt) -- +(135.000000:20.506097pt and 8.485281pt) -- +(225.000000:20.506097pt and 8.485281pt) -- cycle;
\draw (52.500000, 14.000000) node {{$R_1(2\theta_3)$}};
\endscope
\draw[fill=white] (46.500000, -6.000000) rectangle (58.500000, 6.000000);
\draw[very thin] (52.500000, 0.600000) arc (90:150:6.000000pt);
\draw[very thin] (52.500000, 0.600000) arc (90:30:6.000000pt);
\draw[->,>=stealth] (52.500000, -5.400000) -- +(80:10.392305pt);
\draw (77.000000,28.000000) -- (77.000000,14.000000);
\filldraw (77.000000, 28.000000) circle(1.500000pt);
\scope
\draw[fill=white] (77.000000, 14.000000) circle(3.000000pt);
\clip (77.000000, 14.000000) circle(3.000000pt);
\draw (74.000000, 14.000000) -- (80.000000, 14.000000);
\draw (77.000000, 11.000000) -- (77.000000, 17.000000);
\endscope
\scope
\draw[fill=white] (77.000000, -0.000000) +(-45.000000:8.485281pt and 8.485281pt) -- +(45.000000:8.485281pt and 8.485281pt) -- +(135.000000:8.485281pt and 8.485281pt) -- +(225.000000:8.485281pt and 8.485281pt) -- cycle;
\clip (77.000000, -0.000000) +(-45.000000:8.485281pt and 8.485281pt) -- +(45.000000:8.485281pt and 8.485281pt) -- +(135.000000:8.485281pt and 8.485281pt) -- +(225.000000:8.485281pt and 8.485281pt) -- cycle;
\draw (77.000000, -0.000000) node {$X$};
\endscope
\draw (90.000000,42.000000) -- (90.000000,14.000000);
\filldraw (90.000000, 42.000000) circle(1.500000pt);
\scope
\draw[fill=white] (90.000000, 14.000000) circle(3.000000pt);
\clip (90.000000, 14.000000) circle(3.000000pt);
\draw (87.000000, 14.000000) -- (93.000000, 14.000000);
\draw (90.000000, 11.000000) -- (90.000000, 17.000000);
\endscope
\draw[color=black] (95.000000,42.000000) node[right] {$|x_1\rangle$};
\draw[color=black] (95.000000,28.000000) node[right] {$|x_2\rangle$};
\draw[color=black] (95.000000,14.000000) node[right] {$|0\rangle$};
\draw[color=black] (95.000000,0.000000) node[right] {${|0\rangle}$};
\draw[draw opacity=1.000000,fill opacity=0.200000,color=black,dashed] (17.000000,49.000000) rectangle (94.000000,7.000000);
\draw[draw opacity=1.000000,fill opacity=0.200000,color=black,dashed] (17.000000,49.000000) rectangle (94.000000,7.000000);
{\draw (52, 6) -- (52, 7);}
{\draw (53, 6) -- (53, 7);}
\endtikzpicture
  \caption{Example circuit for Construction~\ref{con:anddg-depth-1}
    where $n=2$.  The subscripts on the input qubit label indicate the
    qubit index.}
  \label{fig:anddg-depth-1}
\end{figure}
\begin{example}
  Fig.~\ref{fig:anddg-depth-1} shows the circuit from
  Construction~\ref{con:anddg-depth-1} when $n = 2$.
\end{example}

\section{Summary and Conclusions}
We have described several constructions to implement a \gnot gate controlled by
some Boolean function $f(x)$, without auxiliary qubits or with rotation depth 1,
both for the general case in which the target qubit is in an arbitrary state and
for the special cases in which the target qubit is $|0\rangle$ or
$|f(x)\rangle$.  Q\#~\cite{SGT+18} implementations are available for
Constructions~\ref{con:ccnot-low-qubit},~\ref{con:and-low-qubit},
and~\ref{con:anddg-low-qubit},\footnote{see github.com/microsoft/quantum, sample
\emph{oracle-synthesis}} as well as for Constructions~\ref{con:and-depth-1}
and~\ref{con:anddg-depth-1} when $f$ is the $n$-ary AND function.\footnote{see
github.com/microsoft/quantumlibraries, library functions \texttt{ApplyAnd} and
\texttt{ApplyLowDepthAnd}}

The presented constructions assume worst case complexity, i.e., all
spectral coefficients are nonzero.  Spectral coefficients that are 0
correspond to identity gates in the constructions which may enable
further gate cancellation or savings in auxiliary qubits.  Similarly,
some spectral coefficients lead to rotation gates that are Clifford
gates, and therefore do not contribute to the rotation depth.  It
would be interesting to investigate dedicated constructions that take
such spectral coefficients into account.  Further, an analysis of the
distribution of spectral coefficients over all Boolean functions can
help to better estimate the average complexity of the constructions.

The constructions have exponential worst case complexity either in the
rotation depth, when no auxiliary qubits are allowed, or in the number
of qubits, when the rotation depth is~1.  Pebbling strategies~(see,
e.g.,~\cite{Bennett89,Kralovic01,MSR+19}) can be used to find
tradeoffs in the constructions, e.g., one might fix the number of
auxiliary qubits and then minimize for the rotation depth within this
limit.  Further, quantum circuits might need to obey some layout
constraints due to the qubit coupling in the targeted physical quantum
computer.  Instead of applying additional algorithms as a
post-process, the layout information can be an additional parameter to
the construction algorithm.

\appendix
\subsection{Relationship between spectral coefficients and angles}
\label{sec:app-rotations}
In this section, we prove the relationship between the spectral
coefficients of the control function and the rotation angles that were
used in the constructions above.  Let $f(x_1, \dots, x_n)$ be an
$n$-variable Boolean function, and let
$U_f : |x\rangle|x_{n+1}\rangle \mapsto |x\rangle|x_{n+1} \oplus
f(x)\rangle$, with $x = x_1, \dots, x_n$.  We choose $x_{n+1}$ in
place of $y$ for a more convenient notation in the proof.  Now let
$g(x_1, \dots, x_{n+1}) = x_{n+1} \land f(x_1, \dots, x_n)$ be an
$(n+1)$-variable function, and let
$\hat g = (g_0, \dots, g_{2^{n+1}})^T$ be the function values in
$\{1, -1\}$-coding as in~\eqref{eq:1-1coding}.  We use the following
Lemma from~\cite{SMSM19}.
\begin{lemma}
  Given $f$ and $\hat g$ as above, we have
  \[
    U_f = (I_{2^n} \otimes H) \cdot D \cdot (I_{2^n} \otimes H),
  \]
  where $D = \mathrm{diag}(g_0, \dots, g_{2^{n+1}})$ and $I_{2^n}$ is
  the $2^n\times 2^n$ identity matrix.
\end{lemma}
Using the phase polynomial representation, we can express the action
of $D$ as
\[
  D : |x_1, \dots, x_{n+1}\rangle \mapsto \prod_{k=1}^{2^{n+1}}e^{\mathrm{i}\theta_kp_k(x_1, \dots, x_{n+1})}|x_1, \dots, x_{n_1}\rangle,
\]
where $\theta_k = \frac{\pi s'_k}{2^{n+1}}$ and
$p_k = x_1^{b_1} \oplus \cdots \oplus x_{n+1}^{b_{n+1}}$ when
$k = (b_{n+1}\dots b_1)_2$~\cite{SS03,WGMA14,AM19}.  Here
$(s'_0, \dots, s'_{2^{n+1}-1})^T$ are the spectral coefficients of
$g$.
\begin{theorem}
  Let $s = (s_0, \dots, s_{2^n-1})^T$ be the spectral coefficients of
  $f$.  Then
  \[
    s'_k = 2^n[k \bmod 2^n = 0] + (-1)^{[k \ge 2^n]}s_{k \bmod 2^n},
  \]
  for $0 \le k < 2^{n+1}$, where $[\cdot]$ is the Iverson bracket.
\end{theorem}
\begin{proof}
  Let $\hat f = (f_0, \dots, f_{2^n-1})^T$ as in~\eqref{eq:1-1coding}.
  Since $g = x_{n+1} \land f$, note that
  $\hat g = (1, \dots, 1, f_0, \dots, f_{2^n-1})^T$.  Since
  $s' = H_{n+1}\hat g$ and
  $H_{n+1} = \left(\begin{smallmatrix} H_n & H_n \\ H_n &
      -H_n \end{smallmatrix}\right)$, the upper $2^n$ entries of $s'$
  are $H_n \vec 1 + H_n\hat f$ and the lower $2^n$ entries of $s'$ are
  $H_n \vec 1 - H_n\hat f$, where $\vec 1$ is a $2^n$ column vector in
  which all entries are $1$.  The row sums of $H_n$ are $0$ for all
  but the first row, where the row sum is $2^n$.  Therefore,
  $H_n \vec 1 = (2^n, 0, \dots, 0)^T$.
\end{proof}
The lower entries of $s'$ being negative explains the adjoint
operation of the rotation gates in the constructions.  Further, the
offset of $2^n$ in $s'_{2^n} = 2^ - s_0$ explains the additional $S$
gate in the constructions.

\subsubsection*{Acknowledgements}
We thank Nathan Wiebe, Vadym Kliuchnikov, and Thomas H\"aner for discussions and 
valuable feedback.

\bibliographystyle{IEEEtran}
\bibliography{library}

\end{document}